\documentclass[orivec]{llncs}
\usepackage{amssymb}
\usepackage{amsmath}
\usepackage{graphicx}
\usepackage{acronym}
\usepackage{theorem}
\usepackage{multirow}
\usepackage{verbatim}
\usepackage{booktabs}
\usepackage{color}

\usepackage{url}
\urldef{\mailsa}\path|{m.baldi, f.chiaraluce}@univpm.it, p.santini@pm.univpm.it|
\urldef{\mailsb}\path|{alessandro.barenghi, gerardo.pelosi}@polimi.it|
\newcommand{\keywords}[1]{\par\addvspace\baselineskip
\noindent\keywordname\enspace\ignorespaces#1}

{\theorembodyfont{\rmfamily}
   }
{\theorembodyfont{\rmfamily}
   }
{\theorembodyfont{\rmfamily}
   }
{\theorembodyfont{\rmfamily}
   }
{\theorembodyfont{\rmfamily}
   }

\newcommand{\sysacro}{LEDAkem}

\acrodef{LDPC}{low-density parity-check}
\acrodef{LDGM}{low-density generator matrix}
\acrodef{MDPC}{moderate-density parity-check}
\acrodef{QC}{quasi-cyclic}
\acrodef{QC-LDPC}{quasi-cyclic low-density parity-check}
\acrodef{QC-MDPC}{quasi-cyclic moderate-density parity-check}
\acrodef{RSA}{Rivest-Shamir-Adleman}
\acrodef{BF}{bit flipping}
\acrodef{SPA}{sum product algorithm}
\acrodef{RDF}{random difference families}
\acrodef{ISD}{information set decoding}
\acrodef{KRA}{key recovery attack}
\acrodefplural{KRA}[KRAs]{key recovery attacks}
\acrodef{DA}{decoding attack}
\acrodefplural{DA}[DAs]{decoding attacks}
\acrodef{WF}{work factor}
\acrodef{BER}{bit error rate}
\acrodef{CER}{codeword error rate}
\acrodef{BSC}{binary symmetric channel}
\acrodef{BPSK}{binary phase shift keying}
\acrodef{$2$-PAM}{binary pulse amplitude modulation}
\acrodef{AWGN}{additive white Gaussian noise}
\acrodef{LLR}{log likelihood ratio}
\acrodefplural{LLR}[LLRs]{log likelihood ratios}
\acrodef{SPA}{sum-product algorithm}
\acrodef{DFR}{decryption failure rate}
\acrodef{SL}{security level}
\acrodef{ECC}{elliptic curve cryptography}
\acrodef{QD}{quasi-dyadic}
\acrodef{GRS}{generalized Reed-Solomon}
\acrodef{DSA}{Digital Signature Algorithm}
\acrodef{ECDSA}{Elliptic Curve Digital Signature Algorithm}
\acrodef{KEM}{key encapsulation mechanism}
\acrodefplural{KEM}[KEMs]{key encapsulation mechanisms}
\acrodef{PKC}{public-key cryptosystem}
\acrodef{SK}{secret key}
\acrodef{PK}{public key}
\acrodef{CCA}{chosen ciphertext attack}
\acrodef{IND-CCA}{indistinguishability under chosen ciphertext attack}
\acrodef{IND-CCA2}{indistinguishability under adaptive chosen ciphertext attack}
\acrodef{IND-CPA}{indistinguishability under chosen plaintext attack}
\acrodef{KI}{Kobara-Imai}
\acrodef{PFS}{Perfect Forward Secrecy}
\acrodef{NP}{nondeterministic-polynomial}
\acrodef{DRBG}{deterministic random bit generator}
\acrodef{TRNG}{true random number generator}
\acrodef{KDF}{key derivation function}
\acrodef{AKE}{authenticated key exchange}
\acrodef{KAT}{Known Answer Tests}

\newcommand{\ps}[1]{#1}
\newcommand{\mb}[1]{#1}
\newcommand{\ab}[1]{#1}
   
\begin{document}

\mainmatter  

\title{\sysacro{}: a post-quantum key encapsulation mechanism based on QC-LDPC codes}


%
%
\author{Marco Baldi\inst{1} \and Alessandro Barenghi\inst{2} \and Franco Chiaraluce\inst{1}, \\
Gerardo Pelosi\inst{2} \and Paolo Santini\inst{1}
}
\authorrunning{M. Baldi \and A. Barenghi \and F. Chiaraluce \and G. Pelosi\ \and P. Santini}


\institute{Universit\`a Politecnica delle Marche, Ancona, Italy\\
\mailsa\\
\and
Politecnico di Milano, Milano, Italy\\
\mailsb
}
%
%

\maketitle

\begin{abstract}
This work presents a new code-based key encapsulation mechanism (KEM) called \sysacro{}.
It is built on the Niederreiter cryptosystem and relies on quasi-cyclic low-density parity-check codes as 
secret codes, providing high decoding speeds and compact keypairs. 
\sysacro{} uses ephemeral keys to foil known statistical attacks, 
and takes advantage of a new decoding algorithm that provides faster decoding 
than the classical bit-flipping decoder commonly adopted in this kind of systems.
The main attacks against \sysacro{} are investigated, taking into account quantum speedups.
Some instances of \sysacro{} are designed to achieve different security levels against classical and quantum computers.
Some performance figures obtained through an efficient C99 implementation of \sysacro{} are provided.

\keywords{Code-based cryptography, key encapsulation mechanism, Niederreiter 
cryptosystem, post-quantum cryptography, quasi-cyclic low-density parity-check codes.}
\end{abstract}

\section{Introduction}

Devising efficient and robust post-quantum key encapsulation mechanisms \linebreak (KEMs)\acused{KEM} is an important and urgent research target, as also witnessed by the recent NIST call for post-quantum cryptographic systems \cite{NISTcall2016}.
Code-based\linebreak cryptosystems are among the most promising candidates to replace quantum-vulnerable primitives which are still relying on the hardness of the integer factorization or discrete logarithm problems, such as the Diffie-Hellman key exchange and the \ac{RSA} and ElGamal cryptosystems.
Indeed, Shor's algorithm \cite{Shor1997} can be used to solve both the integer 
factorization and the discrete logarithm problems in polynomial time with a quantum computer.
One of the problems for which no known polynomial time algorithm on a quantum computer exists is the decoding of a general linear code.
Indeed, such a problem belongs to the non deterministic-polynomial (NP)-complete computational equivalence class~\cite{Berlekamp1978,May2011}, which is widely believed to contain problems which have no polynomial time solution on a quantum computer.

The first code-based public-key cryptosystem relying on the general linear code decoding problem was proposed by McEliece in 1978 \cite{McEliece1978}, and used Goppa codes \cite{Goppa1970} to form the secret key.
Such a choice yields large public keys, which is the main limitation of Goppa code-based systems.
The Niederreiter cryptosystem \cite{Niederreiter1986} is a code-based cryptosystem exploiting the same trapdoor, but using syndromes and parity-check matrices instead of codewords 
and generator matrices as in McEliece.
When the same family of codes is used, Niederreiter and McEliece 
are equivalent \cite{Li1994} and therefore they achieve the same security levels.

Replacing Goppa codes with other families of more structured codes may reduce the public key size.
However, this may also compromise the system security, as it occurred with some first McEliece variants based on \ac{QC} codes \cite{Gaborit2005}, \ac{LDPC} codes \cite{Monico2000} and \ac{QC-LDPC} codes \cite{Otmani2008}, \ac{QD} codes \cite{Misoczki2009}, convolutional codes \cite{Londahl2012} and some instances based on \ac{GRS} codes \cite{Baldi2016joc,Berger2005}.
Nevertheless, some variants exploiting \ac{QC-LDPC} and \ac{QC-MDPC} codes \cite{Baldi2013,Baldi2008,Misoczki2013} have been shown to be able to achieve very compact keys without endangering security.

Recently, some new statistical attacks have been developed that exploit the information coming from decryption failures in \ac{QC-LDPC} and \ac{QC-MDPC} code-based systems to perform key recovery attacks, thus forcing to renew keys frequently in these systems~\cite{Fabsic2017,Guo2016}.

In this paper, we start from the \ac{QC-LDPC} code-based system proposed in \cite{Baldi2008,Baldi2013} and we develop a new \ac{KEM} based on the the Niederreiter cryptosystem.
We also introduce an improved decoding algorithm which exploits correlation among intentional errors seen by the private code. This way, the correction capability of the private code is exploited to the utmost, thus allowing to achieve significant reductions in the public key size.
We call the new system \sysacro{} and study its properties and security.
We take into account the fact that Grover's algorithm running on a quantum computer may be exploited to speedup attacks based on \ac{ISD} \mb{\cite{KT17,Vries2016}}, and we propose some sets of parameters for \sysacro{} achieving different security levels against attacks exploiting both classical and quantum computers.
We also describe an optimized software implementation of the proposed system and provide and discuss some performance figures.
\mb{\sysacro{} currently is one of the first round candidate algorithms of the NIST 
post-quantum cryptography standardization project \cite{NISTcall2016}, along with other 
code-based \acp{KEM}. In this work we will highlight the differences between our proposal
and the closest one among the others, i.e. BIKE \cite{BIKE2017}, which relies on \ac{QC-MDPC}
codes for its construction.}

The organization of the paper is as follows. In Section \ref{sec:description} we 
describe \sysacro{}. In Section \ref{sec:Security} we present its security 
analysis and in Section \ref{sec:Properties} its peculiar features. In 
Section \ref{sec:results} we discuss some implementation issues and we show some 
numerical results. Finally, some conclusions are drawn in Section \ref{sec:conclusion}.

\section{The \sysacro{} cryptosystem}
\label{sec:description}

The \sysacro{} cryptosystem is derived from the Niederreiter cryptosystem with the following main differences:
\begin{itemize}
\item Non-algebraic codes known as \ac{QC-LDPC} codes are used as secret codes.
\item The public code is neither coincident with nor equivalent to the private code.
\item Suitably designed iterative non-bounded-distance decoding algorithms are used.
\end{itemize}

The motivation for using \ac{QC-LDPC} codes as private codes is in the fact that 
these such codes are known to achieve important reductions in the public key size when 
used in this context \cite{Baldi2013,Misoczki2013}.
Moreover, when \ac{LDPC} codes are used as private codes, the public code cannot 
be either coincident with or equivalent to the private code.
Indeed, in such a case, an attacker could search for low weight codewords in the 
dual of the public code and find a sparse parity-check matrix of the private 
code which allows efficient decoding.

For this reason, following \cite{Baldi2013}, \sysacro{} uses a transformation 
matrix $Q$ that hides the sparse parity-check matrix $H$ of the private code 
into a denser parity-check matrix $L = HQ$ of the public code.
This also affects the error vector that must be corrected during decryption, 
which is obtained from the error vector used during encryption through 
multiplication by $Q$.
In this work, we show how it is possible to exploit the knowledge of $Q$ 
to design an ad-hoc decoding algorithm achieving very good performance in terms 
of both decoding speed and \ac{DFR}.

In fact, a well-known feature of \ac{LDPC} coding is that the decoding radius of iterative decoders is not sharp and cannot be estimated in a deterministic way.
It follows that some residual \ac{DFR} must be tolerated, and it must be estimated heuristically through Montecarlo simulations.
This is done for all the proposed instances of \sysacro{} in order to guarantee that they achieve a sufficiently low \ac{DFR}.
Providing quantitative estimates of the \ac{DFR} for the proposed instances
of \sysacro{} allows us to prevent attacks such as the ones described in~\cite{Fabsic2017,Guo2016}
changing the key either at each round of the \ac{KEM}, or before a sufficient amount of decoding failures are observed by the attacker.

\subsection{Coding background}
\label{subsec:prel}
A \ac{QC} code is defined as a linear block code with dimension $k=p k_{0}$ 
and length $n=p n_{0}$, in which each cyclic shift of a codeword by $n_{0}$ 
symbols results in another valid codeword.
It follows from their definition that \ac{QC} codes have generator and 
parity-check matrices in ``blocks circulant'' form or, equivalently, in ``circulants block'' form.
The latter is used in \sysacro{}. 
A $v\times v$ circulant matrix $A$ has the following form
\begin{equation}
A = \left[{\begin{array}{ccccc}
{a_{0}} & {a_{1}} & {a_{2}} & \cdots & {a_{v-1}}\\
{a_{v-1}} & {a_{0}} & {a_{1}} & \cdots & {a_{v-2}}\\
{a_{v-2}} & {a_{v-1}} & {a_{0}} & \cdots & {a_{v-3}}\\
\vdots & \vdots & \vdots & \ddots & \vdots\\
{a_{1}} & {a_{2}} & {a_{3}} & \cdots & {a_{0}}
\end{array}}\right].\label{eq:CircMatrix}
\end{equation}
According to its definition, any circulant matrix is regular, since all its rows
and columns are cyclic shifts of the first row and column, respectively.

The set of $v\times v$ binary circulant matrices forms an algebraic ring under
the standard operations of modulo-$2$ matrix addition and multiplication.
The zero element is the all-zero matrix, and the identity element is the $v\times v$ identity matrix.
The algebra of the polynomial ring $\mathbb{F}_2[x]/\langle x^v +1 \rangle$ is isomorphic to the ring of $v\times v$ circulant matrices over $\mathbb{F}_2$ with the following map
\begin{equation}
A \leftrightarrow a\left(x\right)=\sum_{i=0}^{v-1}a_{i} x^{i}.
\label{eq:CircPoly}
\end{equation}
According to \eqref{eq:CircPoly}, any binary circulant matrix is associated to 
a polynomial in the variable $x$ having coefficients over $\mathbb{F}_2$ which coincide with the entries of the first row of the matrix
\begin{equation}
a\left(x\right)=a_{0}+a_{1}x+a_{2}x^{2}+a_{3}x^{3}+\cdots+a_{v-1}x^{v-1}.\label{eq:CircPoly2}
\end{equation}
According to \eqref{eq:CircPoly}, the
all-zero circulant matrix corresponds to the null polynomial and the identity 
matrix to the unitary polynomial.
The ring of polynomials $\mathbb{F}_2[x]/\langle x^v +1 \rangle$ includes 
elements that are zero divisors which are mapped to singular circulant matrices 
over $\mathbb{F}_2$.
Avoiding such matrices is important in some parts of \sysacro{}, and smart ways 
exist to design non-singular circulant matrices.
As it will be described next, the main part of the secret key of \sysacro{} is 
formed by a binary \ac{QC-LDPC} code described through its parity-check matrix 
$H$.
Let $n$ denote the code length in bits and $k$ denote the code dimension in 
bits, then $H$ has size $(n-k) \times n = r \times n$, where $r$ is the code redundancy.

\subsection{Description of the primitives}

The main functions of \sysacro{} are described next.

\subsubsection{Key generation. \label{sec:KeyGen}}

Both private and public keys consist of binary matrices.
These matrices, in their turn, are formed by $p \times p$ circulant blocks, being $p$ an integer properly chosen.

\paragraph{Secret key.}

The key generation input is formed by:
\begin{itemize}
\item The circulant block size $p$ (usually in the order of some thousands bits).
\item The integer $n_0$ (usually between $2$ and $4$), representing the number of circulant blocks forming the matrix $H$.
\item The integer $d_v$, representing the row/column weight (usually between $15$ and $25$) of the circulant blocks forming the matrix $H$.
\item The vector of integers $\bar{m} = \left[m_0, m_1, \ldots, m_{n_0-1}\right]$, representing the row/column weights (each entry usually smaller than $10$) of the circulant blocks forming the matrix $Q$ (the structure of $Q$ is clarified below).
\end{itemize}
Given these inputs, the secret key is obtained as follows.

First, $n_0$ sparse circulant matrices with size $p \times p$ are generated at random.
Each of them has row/column weight $d_v$.
We denote such matrices as $H_0, H_1, \ldots, H_{n_0-1}$.
The secret low-density parity-check matrix $H$ is then obtained as
\begin{equation}
H = \left[H_0 | H_1 | H_2 | \ldots | H_{n_0-1} \right].
\label{eq:Hsecret}
\end{equation}

The size of $H$ is $p \times n_0 p$.
Other $n_0^2$ sparse circulant blocks $Q_{i,j}$ are then randomly generated to form the secret sparse matrix
\begin{equation}
Q = \left[
\begin{array}{cccc}
Q_{0,0} 		&	Q_{0,1} 	& \ldots & Q_{0,{n_0-1}}     \\
Q_{1,0} 		&   Q_{1,1} 	& \ldots & Q_{1,{n_0-1}}     \\
\vdots 		    & \vdots 		& \ddots & \vdots          \\
Q_{{n_0-1},0} 	& Q_{{n_0-1},1} 	& \ldots & Q_{{n_0-1},{n_0-1}} \\
\end{array}
\right].
\label{eq:Qsecret}
\end{equation}
The row/column weight of each block $Q_{i,j}$ is fixed according to the following matrix:
\begin{equation}
w(Q) = \left[
\begin{array}{cccc}
m_0 		&	m_1 		& \ldots & m_{{n_0-1}} \\
m_{n_0-1}	&   m_0 		& \ldots & m_{{n_0-2}} \\
\vdots 	    & \vdots 	    & \ddots & \vdots    \\
m_1 		& m_{{n_0-1}}	    & \ldots & m_0       \\
\end{array}
\right],
\end{equation}
such that each row and each column of $Q$ has weight $m = \sum_{i=0}^{n_0-1}{m_i}$.

The choice of the weights $\bar{m}=[m_0,m_1,\cdots, m_{n_0-1}]$ and the size $p$ 
of the circulant blocks composing it is very important since it allows to 
discern if $Q$ is invertible or not. 
In particular, denoting with $\mathbf{\Pi}\left\{\cdot\right\}$ the permanent of 
a matrix, the following theorem holds.

\begin{theorem}\label{the:InvQ}{Let $p>2$ be a prime such that 
$\text{ord}_p(2) = p-1$ and $Q$ be an $n_0\times n_0$ matrix of elements in 
$\mathbb{F}_2[x]/\langle x^p+1\rangle$; 
if $\mathbf{\Pi}\left\{w(Q)\right\}$ is odd and $\mathbf{\Pi}\left\{w(Q)\right\}<p$, 
then $Q$ is non singular.} 
\end{theorem}
\begin{proof}
{\it Omitted for the sake of brevity.}
\end{proof}
With this result, we can guarantee that, when the sequence $\bar{m}$ is properly 
chosen, the matrix $Q$ is always non singular, which is a necessary condition for 
the key generation process to be successful.

\begin{definition}
The \ac{SK} of \sysacro{} is formed by $\left\{H, Q\right\}$.
\end{definition}

Since both $H$ and $Q$ are formed by sparse circulant blocks, it is convenient to 
represent each of these blocks through the indexes of the symbols $1$ in their 
first row, i.e. adopt a sparse representation for them.
Each index of this type requires $\left\lceil \log_2(p) \right\rceil$ bits to 
be stored.
If we consider that the circulant blocks in any block row of $Q$ have overall 
weight $m = \sum_{i=0}^{n_0-1}{m_i}$, the size of \ac{SK} in bits is
\begin{equation}
S_{sk} = n_0 \left(d_v + m \right)  \left\lceil \log_2(p) \right\rceil.
\label{eq:Ssk}
\end{equation}

In practice, the secret matrices are generated through a \ac{DRBG}, seeded with 
a bit string extracted from a \ac{TRNG}. 
In this case, to obtain $H$ and $Q$ it is sufficient to know the \ac{TRNG} 
extracted seed of the \ac{DRBG} that has been used to generate the positions of 
their non-null coefficients, since this process is rather fast.
This approach allows reducing the size of the secret key to the minimum required,
as it is assumed that the \ac{TRNG} output cannot be compressed.
The entity of the reduction depends on the values of the parameters involved in 
\eqref{eq:Ssk}.
\paragraph{Public key.}

Starting from $H$ and $Q$, the following binary matrices are computed.
First of all, the matrix $L$ is obtained as
\begin{equation}
L = HQ = \left[L_0 | L_1 | L_2 | \ldots | L_{n_0-1} \right].
\label{eq:L}
\end{equation}

If both $d_v$ and $m$ are odd, then $L_{n_0-1}$ has full-rank.
In fact, $L_{n_0-1}=\sum_{i=0}^{n_0-1}{H_iQ_{i,{n_0-1}}}$ and has weight equal 
to $md_v-2c$ (where $c$ is the number of cancellations occurred in the product).
It is possible to demonstrate that if $md_v$ is odd and $md_v<p$ then 
$L_{n_0-1}$ is non-singular.

After inverting $L_{n_0}$, the following matrix is computed:
\begin{equation}
M = L_{n_0-1}^{-1} L = \left[M_0 | M_1 | M_2 | \ldots | M_{n_0-2} | I \right] = \left[M_l | I \right].
\end{equation}

\begin{definition}
The \ac{PK} of \sysacro{} is formed by $M_l = \left[M_0 | M_1 \right.$ $\left.| M_2 | \ldots | M_{n_0-2} \right]$.
\end{definition}

Since the circulant blocks forming $M_l$ are dense, it is convenient to store 
them through the binary representation of their first row (the other rows are 
then obtained as cyclic shifts of the first row).
The bit-size of the \ac{PK} hence is
\begin{equation}
S_{pk} = \left(n_0-1\right) p.
\end{equation}

\subsubsection{Encryption. \label{sec:Enc}}

The plaintext of \sysacro{} is an ephemeral random secret generated by Bob
who is willing to share it with Alice.
The encryption inputs are:
\begin{itemize}
\item The values of $n_0$ and $p$, from which $n = n_0 p$ is computed.
\item The number of intentional errors $t \ll n$.
\end{itemize}

Bob generates a secret in the form of a random binary vector $e$ with length of 
$n = n_0p$ bits and Hamming weight $t$.
Given a \ac{KDF}, the shared secret key $k_s$ is generated from $e$ as 
$k_s = \text{\sc KDF}(e)$. 
In order to encapsulate the shared secret $e$, Bob fetches Alice's \ac{PK} 
$M_l$ and computes $s = \left[M_l | I \right] e^T$
where $^T$ denotes matrix transposition. The $p \times 1$ syndrome vector $s$ 
representing the encapsulated secret is then sent to Alice.

\subsubsection{Decryption. \label{sec:Dec}}

In order to perform decryption, Alice must recover $e$ from $s$.
The latter can be written as $s = M e^T = L_{n_0-1}^{-1} L e^T = L_{n_0-1}^{-1} H Q e^T.$
The first decryption step for Alice is computing $s' = L_{n_0-1} s = H Q e^T$.
For this purpose, Alice needs to know $L_{n_0-1}$ that, 
according to \eqref{eq:L}, is the last circulant block of the matrix $HQ$. 
Hence, it can be easily computed from the \ac{SK} which contains both $H$ and $Q$.
If we define the \textit{expanded error vector} as
\begin{equation}
e' = e Q^T,
\end{equation}
then we have $s' = H e'^T$.
Hence, \ac{QC-LDPC} decoding through $H$ can be exploited for recovering $e'$ 
from $s'$. \ac{QC-LDPC} decoders are not bounded distance decoders, and some 
\ac{DFR} must be tolerated.
However, the system parameters can be chosen such that the \ac{DFR} is 
acceptably small.
For this purpose, the average decoding radius of the private code must be 
sufficiently larger than the Hamming weight of $e'$, which is approximately 
equal to $mt$ (due to the sparsity of $Q$ and $e$).
Then, multiplication by $\left(Q^T\right)^{-1}$ would be needed to obtain $e$ 
from $e'$, that is,
\begin{equation}
e = e' \left(Q^T\right)^{-1}.
\end{equation}

However, by exploiting the efficient decoding algorithm described in Section~\ref{subsec:Qdec}, this last step can be avoided, which also allows avoiding the computation and storage of $\left(Q^T\right)^{-1}$ as part of the 
secret key.
In fact, the decoding algorithm described in Section~\ref{subsec:Qdec} allows recovering $e$ directly by performing decoding of $s' = L_{n_0-1} s = H Q e^T$ 
through $H$, while taking into account the effect of the multiplication of $e$ by $Q$.
Then, the secret key is recovered as $k_s = \text{\sc KDF}(e)$.

In case a decoding error occurs, the decryption procedure derives the shared secret combining with a KDF the syndrome with a secret constant, which may be derived via a PRNG from the secret key material~\cite{cryptoeprint:2017:1005}. \mb{Alternatively, using a secret permutation of the syndrome as input to the KDF was noted to be effective in~\cite{cryptoeprint:2017:604}. Such an approach which is beneficial from the security standpoint in case of an accidental keypair reuse. More details concerning this aspect, which is related to formal security of \sysacro{}, will be given in Section \ref{sec:Properties}.}
According to this approach, Bob will become aware of the decoding failure upon reception of the message sent by Alice encrypted with the incorrectly derived shared secret.

\subsection{Efficient decoding \label{subsec:Qdec}}

Classical \ac{BF} decoding works as follows.
At each iteration, for each codeword bit position, the number of unsatisfied 
parity-check equations is computed, and if this number equals or exceeds a given 
threshold, then that bit is flipped.
The decision threshold can be chosen in many ways, affecting the decoder 
performance, and it can be fixed or it can vary during iterations.
A choice that often turns out to be optimal is to fix the threshold, at each 
iteration, as the maximum number of unsatisfied parity-check equations in which 
any codeword bit is involved.
In fact, a codeword bit participating in a higher number of unsatisfied 
parity-check equations can be considered less reliable than a codeword bit 
participating in a smaller number of unsatisfied parity-check equations.
So, if the threshold is chosen in this way, the bits that are flipped are those 
that are most likely affected by errors.

Starting from classical \ac{BF}, we have developed an improved decoder that is 
specifically designed for \sysacro{}, where the position of the ones in the 
expanded error vector $e'$ to be corrected is influenced by the value of $Q^T$, 
as $e'$ is equivalent to a random error vector $e$ with weight $t$ multiplied by 
$Q^{T}$.
Since this improved decoder takes into account such a multiplication by the 
transpose of matrix $Q$ to estimate with greater efficiency
the locations of the bits of the expanded error vector, we denote it as 
\textit{Q-decoder}.

Inputs of the decoder are the syndrome $s'$ and 
the matrices $H$ and $Q$ according to \eqref{eq:Hsecret} and \eqref{eq:Qsecret}, 
respectively. 
The output of the decoder is a $1 \times n$ vector $\hat{e}$ or a decoding 
failure, where $\hat{e}$ represents the decoder estimate of the error vector 
$e$ appearing in the equality $s'= H Q e^T$.
The decoding process performs a maximum of $l_{max}$ iterations, where the $l$-th iteration processes $s^{(l-1)}$ and $\hat{e}^{(l-1)}$ (that is the values at the previous iteration) and outputs $s^{(l)}$ and 
$\hat{e}^{(l)}$.
\ps{A threshold criterion is adopted to compute the positions in $\hat{e}^{(l)}$ that must be changed.
The threshold values $b^{(l)}$ can be chosen in different ways and affect the decoder performance. In the next section we describe a simple and effective procedure to design such values.}
The decoder initialization is performed by setting $s^{(0)}= s'^T$ and 
$\hat{e}^{(0)}=0_n$, where $0_n$ is the length-$n$ vector with all-zero entries.
\mb{It is important to note that $s^{(0)}$ (and, by extension, $s^{(l)}$) is a row vector. Moreover, let us consider that all multiplications are binary, expect those denoted with `$*$', which are performed in the integer domain $\mathbb{Z}$.} The $l$-th iteration of the Q-decoder performs the following operations:
\begin{enumerate}
\item Compute \mb{$\Sigma^{(l)}=\left[\sigma_1^{(l)},\sigma_2^{(l)},\cdots,\sigma_n^{(l)}\right]=s^{(l-1)}* H$, resulting in a vector of integers having entries between $0$ and $d_v$}.
\item Compute \mb{$R^{(l)}=\left[\rho_1^{(l)},\rho_2^{(l)},\cdots,\rho_n^{(l)}\right]=\Sigma^{(l)} * Q$}.
\item \ps{Define $\Im^{(l)}=\left\{v \in [1,n]| \hspace{1mm} \rho^{(l)}_{v}\geq b^{(l)}\right\}$.}
\item \ps{Update $\hat{e}^{(l-1)}$ as
\begin{center}
$\hat{e}^{(l)}=\hat{e}^{(l-1)}+1_{\Im^{(l)}}$
\end{center}
where $1_{\Im^{(l)}}$ is a length-$n$ binary vector with all-zero entries, except those indexed by $\Im^{(l)}$.}
\item \ps{Update the syndrome as
\begin{center}
$s^{(l)}=s^{(l-1)}+\sum_{v\in \Im^{(l)}}{q_v} H^T$
\end{center}
where $q_v$ is the $v$-th row of $Q^T$.}
\item If the weight of $s^{(l)}$ is zero then stop decoding and return $\hat{e}^{(l)}$.
\item If $l < l_{max}$ then increment $l$ and go back to step i), 
otherwise stop decoding and return a decoding failure. 
\end{enumerate}

As in classical \ac{BF}, the first step of this algorithm computes the vector $\Sigma^{(l)}$. 
Each entry of this vector counts the number of unsatisfied parity-check 
equations corresponding to that bit position, and takes values in 
$\{0,\ldots, d_v\}$.
This evaluates the likelihood that the binary element of $e'$ at the same position is equal to one.
Differently from classical \ac{BF}, in step ii) the correlation $R^{(l)}$ 
between these likelihoods and the rows of $Q^T$ is computed.
In fact, the expanded error vector $e'=eQ^T$ can be written as the sum of the 
rows of $Q^T$ indexed by the support of $e$, that is 
$e'=\sum_{j\in \Psi\left\{e\right\}}{q_j}$
where $\Psi\left\{e\right\}$ denotes the support of $e$. 

Since both $Q$ and $e$ are sparse (that is, $m,t \ll n$), cancellations between 
ones in the sum are very unlikely.
When the correlation between $\Sigma^{(l)}$ and a generic row $q_v$ of $Q^T$ 
is computed, two cases may occur:
\begin{itemize}
\item If $v\notin\Psi\left\{e\right\}$, then it is very likely that $q_v$ has 
a very small number of common ones with all the rows of $Q^T$ forming $e'$,
hence the correlation is small.
\item If $v\in\Psi\left\{e\right\}$, then $q_v$ is one of the rows of $Q^T$ 
forming $e'$, hence the correlation is large.
\end{itemize}

The main difference with classical \ac{BF} is that, while in the latter all 
error positions are considered as independent, the Q-decoder exploits the 
correlation among expanded errors which is present in \sysacro{}, since their positions are influenced by $Q^T$.
This allows achieving important reductions in the number of decoding iterations.
As a further advantage, this decoder allows recovering $e$, besides $e'$, 
without the need of computing and storing the inverse of the matrix $Q^T$.
For this purpose, it is sufficient that, at each iteration, the Q-decoder flips the bits of the estimated error vector $e$ that correspond to the correlations values overcoming the threshold.

\subsection{Choice of the Q-decoder decision thresholds}
\label{sec:choice}

One important aspect affecting performance of the Q-decoder is the choice of the 
threshold values against which the correlation is compared at each iteration.
A natural choice is to set the threshold used at iteration $l$ equal to the maximum value of the correlation $R^{(l)}$ \ps{, that is $b^{(l)}=\max_{j=1,2,\cdots,n}{\left\{\rho_j^{(l)}\right\}}$}.
This strategy ensures that only those few bits that have maximum likelihood of 
being affected by errors are flipped during each iteration, thus achieving the 
lowest \ac{DFR}.
However, such an approach has some drawbacks in terms of complexity, since the computation of 
the maximum correlation requires additional computations with respect to a 
fixed threshold.

Therefore, \ps{as in \cite{Chaulet2016},} we consider a different strategy, which allows computing the 
threshold values on the basis of the syndrome weight at each iteration.
According to this approach, during an iteration it is sufficient to compute the 
syndrome weight and read the corresponding threshold value from a look-up table.
This strategy still allows to achieve a sufficiently low \ac{DFR}, while employing
a significantly smaller number of decoding iterations.

Let us consider the $l$-th iteration of the Q-decoder, and denote by $t_l$ the 
weight of the error vector $e^{(l)}$ and with $t'_l$ the weight of the 
corresponding expanded error vector $e'^{(l)} = e^{(l)}Q^T$.
Let us introduce the following probabilities \cite{Baldi2012}:
\begin{align}
p_{ci}(t'_l) &= \sum_{\text{$j=0$, $j$ odd}}^{\min\left[n_0d_v-1,t'_l\right]}{\frac{\binom{n_0d_v-1}{j}\binom{n-n_0d_v}{t'_l-j}}{\binom{n-1}{t'_l}}} \nonumber \\
p_{ic}(t'_l) &= \sum_{\text{$j=0$, $j$ even}}^{\min\left[n_0d_v-1,t'_l-1\right]}{\frac{\binom{n_0d_v-1}{j}\binom{n-n_0d_v}{t'_l-j-1}}{\binom{n-1}{t'_l-1}}}
\label{eq:probabilities}
\end{align}
where $p_{ci}(t'_l)$ is the probability that a codeword bit is error-free and a 
parity-check equation evaluates it to be incorrect, and $p_{ic}(t'_l)$ is the probability that a codeword bit is error-affected and a 
parity-check equation evaluates it to be correct.
In both these cases, the syndrome bit is equal to $1$.
The probability that each syndrome bit is equal to $1$ can be therefore computed 
as $p_{ic}(t'_l)+p_{ci}(t'_l)$, so the average syndrome weight at iteration $l$ results in
\begin{equation}
w_s^{(l)}=E\left[wt\left\{s^{(l)}\right\}\right]=\left[p_{ic}(t'_l)+p_{ci}(t'_l)\right]p
\label{eq:Ave_ws}
\end{equation}  
where $wt\left\{\cdot\right\}$ denotes the Hamming weight.
Since both the parity-check matrix and the error vector are sparse, the 
probability of $wt\left\{s^{(l)}\right\}$ being significantly different from 
$w_s^{(l)}$ is negligible.

So, \eqref{eq:Ave_ws} allows predicting the average syndrome weight starting 
from $t'_l$.
In order to predict how $t'_l$ varies during iterations, let us consider the 
$i$-th codeword bit and the corresponding correlation value $\rho_i^{(l)}$ at the $l$-th
iteration.
The probability that such a codeword bit is affected by an error can be written as
\begin{equation}
P\left\{e_i=1|\rho_i^{(l)}\right\}=\frac{P\left\{e_i=1,\rho_i^{(l)}\right\}}{P\left\{\rho_i^{(l)}\right\}}=\left(1+\frac{P\left\{e_i=0,\rho_i^{(l)}\right\}}{P\left\{e_i=1,\rho_i^{(l)}\right\}}\right)^{-1}
\label{eq:Pcond}
\end{equation}
where $e_i$ is the $i$-th bit of the error vector used during encryption.
After some calculations, we obtain
\begin{equation}
P\left\{e_i=1|\rho_i^{(l)}\right\}=\frac{1}{1+\frac{n-t_l}{t_l}\left(\frac{p_{ci}(t_l)}{p_{ic}(t_l)}\right)^{\rho_i^{(l)}}\left(\frac{1-p_{ci}(t_l)}{1-p_{ic}(t_l)}\right)^{md_v-\rho_i^{(l)}}}
\label{eq:Pei1}
\end{equation}
where $p_{ci}(t_l)$ and $p_{ic}(t_l)$ are given in \eqref{eq:probabilities}, 
with $t_l$ as argument instead of $t'_l$.

Adding the $i$-th row of $Q^T$ to the expanded error vector $e'$ is the same as 
flipping the $i$-th bit of the error vector $e$.
Hence, we can focus on $e$ and on how its weight $t_l$ changes during decoding 
iterations.
The values of $t'_l$ can be estimated using \eqref{eq:Ave_ws}, while, due to sparsity, those of $t_l$ can be estimated as $t'_l/m$.

The decision to flip the $i$-th codeword bit is taken when the following 
condition is fulfilled
\begin{equation}
P\left\{e_i=1|\rho_i^{(l)}\right\}>(1+\Delta)P\left\{e_i=0|\rho_i^{(l)}\right\}
\label{eq:PeiCondition}
\end{equation}
where $\Delta \geq 0$ represents a margin that must be chosen taking into account 
the \ac{DFR} and complexity: increasing $\Delta$ decreases the \ac{DFR} but 
increases the number of decoding iterations.
So, a trade-off value of $\Delta$ can be found that allows achieving a 
low \ac{DFR} while avoiding unnecessary large numbers of iterations.

Since $P\left\{e_i=0|\rho_i^{(l)}\right\}=1-P\left\{e_i=1|\rho_i^{(l)}\right\}$, \eqref{eq:PeiCondition} can be rewritten as
\begin{equation}
P\left\{e_i=1|\rho_i^{(l)}\right\}>\frac{1+\Delta}{2+\Delta}.
\label{eq:BoundFlip}
\end{equation}
$P\left\{e_i=1|\rho_i^{(l)}\right\}$ is an increasing function of 
$\rho_i^{(l)}$, hence the minimum value of $\rho_i^{(l)}$ such that 
\eqref{eq:BoundFlip} is satisfied can be computed as
\begin{equation}
b^{(l)}=\min{\left\{\rho_i^{(l)} \in [0, md_v], \quad {\rm s.t.} \quad P\left\{e_i=1|\rho_i^{(l)}\right\}>\frac{1+\Delta}{2+\Delta}\right\}}
\label{eq:Optbl}
\end{equation}
and used as the decision threshold at iteration $l$.

Based on the above considerations, the procedure to compute the decision threshold 
value per each iteration as a function of the syndrome weight can be summarized 
as follows:
\begin{enumerate}
\item The syndrome weights corresponding to $t'_l=0,m,2m,\cdots,mt$ (which are 
all the possible values of $t'_l$ neglecting cancellations) are computed 
according to \eqref{eq:Ave_ws}. These values are denoted as 
$\left\{w_s(0), w_s(m), \cdots, w_s(mt)\right\}$.
\item At iteration $l$, given the syndrome weight $\bar{w_s}^{(l)}$, the integer 
$j\in [0,t]$ such that $w_s(jm)$ is as close as possible to $\bar{w_s}^{(l)}$ is computed.
\item Consider $t_l = j$ and compute $b^{(l)}$ according to \eqref{eq:Optbl} and 
\eqref{eq:Pei1}. The value of $b^{(l)}$, so obtained, is used as the decoding 
threshold for iteration $l$.
\end{enumerate}

The above procedure can be implemented efficiently by populating a look-up table 
with the pairs $\left \{ w_j, b_j \right \}$, sequentially ordered.
During an iteration, it is enough to compute $\bar{w_s}^{(l)}$, search the largest
$w_j$ in the look-up table such that $w_j<\bar{w_s}^{(l)}$ and set $b^{(l)}=b_j$.

We have observed that, moving from large to smalle values of $w_j$, the thresholds computed this way firstly exhibit a decreasing trend,
then start to increase.
According to numerical simulations, neglecting the final increase is beneficial 
from the performance standpoint. Therefore, in the look-up table we replace the 
threshold values after the minimum with a constant value equal to the minimum 
itself.
\mb{\subsection{Relations with QC-MDPC code-based systems
\label{sec:qcmdpcSecurity}}}
\mb{In \sysacro{}, the public code is a \ac{QC-MDPC} code that admits $L=HQ$ as a valid parity-check matrix. 
However, differently from \ac{QC-MDPC} code-based schemes, the private code is a \ac{QC-LDPC} code, which facilitates decoding.
In fact, decoding directly the public \ac{QC-MDPC} code through classical \ac{BF} decoders would be a possibility, but the approach we follow is different.
By using the decoding algorithm described in Section~\ref{subsec:Qdec}, we decode the private \ac{QC-LDPC} code, taking into account the correlation introduced in the private error vector due to multiplication by $Q^T$.
Since the private \ac{QC-LDPC} matrix is sparser than the \ac{QC-MDPC} matrix of the public code, this yields lower decoding complexity.}

\mb{Besides working over different matrices, the main difference between these two decoding algorithms is in the use of integer multiplications in our decoder, while all multiplications are performed over $\mathbb{F}_2$ in classical \ac{BF} decoders.
In fact, in our decoder we perform the following operation to compute $R^{(l)}$
\begin{equation}
\label{eq_R}
R^{(l)}=s^{(l-1)} * H * Q= e Q^T H^T * H * Q \approx e L^T * L
\end{equation}
where the last approximation comes from the fact that, for two sparse matrices $A$ and $B$, we have $A \cdot B \approx A * B$.
Thus, we can say that $H Q \approx H * Q$.
So, if we consider classical \ac{BF} decoding working over the matrix $L = H Q$, the counter vector is computed as
\begin{equation}
\label{eq_He}
\Sigma^{(l)} = s^{(l-1)} * L = e L^T * L.
\end{equation}
}

\mb{
In the Q-decoder, the error vector is updated by summing rows of $Q^T$, which is equivalent to flipping bits of the public error vector.
Hence, there is a clear analogy between decoding of the private \ac{QC-LDPC} code through the Q-decoder and decoding of the public \ac{QC-MDPC} code through a classical \ac{BF} decoder.
Through numerical simulations we have verified that the two approaches yield comparable performance in the waterfall region.
Performance in the error floor region is instead dominated by the minimum distance of the code over which decoding is performed.
Since \ac{QC-LDPC} codes have smaller minimum distance than \ac{QC-MDPC} codes, this reflects into a higher error floor when decoding is performed over the private \ac{QC-LDPC} code.
However, no error floor has been observed during simulations of \sysacro{} with \ac{QC-LDPC} decoding, down to a \ac{DFR} between $10^{-9}$ and $10^{-8}$.
Since this is the working point of the codes we use, in terms of \ac{DFR}, we can say that the error floor effect, if present, is negligible from our scheme performance standpoint.\\
$\ $
}

\section{Security analysis}
\label{sec:Security}

\ab{\sysacro{} is constructed starting from the computational problem of syndrome decoding, i.e., 
obtaining a bounded weight error vector from a given syndrome and a general 
linear code, which was shown to be NP-complete in~\cite{Berlekamp1978}.
The main difference from the statement of the general hard problem on which
our proposal is built is the nature of the code employed, which is quasi-cyclic
and admits a representation with a low-density parity-check matrix.
To best of our knowledge, there is no superpolynomial advantage in performing syndrome
decoding on QC-LDPC, given our public code representation, either due to the 
quasi-cyclic form of the code or to the low density of its parity matrix.
We point out that the same assumption on the lack of advantage due to the quasi-cyclic structure of a code has also been done in
both the BIKE~\cite{BIKE2017} and the BIG QUAKE~\cite{BIGQUAKE2017} proposals.}
With these statements standing, the security analysis of \sysacro{} examines and 
quantifies the effectiveness of the best known attacks detailing the efficiency 
of algorithms running on both classical and quantum computers providing 
non-exponential speedups over an enumerative search for the correct error vector.
We remark that currently no algorithm running on either a classical Turing 
Machine (TM) or a quantum TM provides an exponential speedup in solving the 
computational problem underlying  \sysacro{} compared to an exhaustive search 
approach.

\subsection{Analysis of the algorithm with respect to known attacks}
As mentioned in the previous sections, \sysacro{} derives from \ac{QC-LDPC} 
code-based cryptosystems already established in the 
literature~\cite{Baldi2012,Baldi2014book}. 
As proved in \cite{Fabsic2017}, in case of using long-term keys, these 
cryptosystems may be subject to reaction attacks that are able to recover the 
secret key by exploiting the inherent non-zero \ac{DFR} they exhibit and Bob's 
reactions upon decryption failures. 
However, using ephemeral keys prevents the possibility to mount an attack of 
this kind, which requires long statistical evaluations.
Nevertheless, the risks in case of an accidental keypair reuse must be considered, and this will be done in Section \ref{sec:Properties}.

A first type of attacks that can be mounted against \sysacro{} are \acp{DA} 
aimed at performing decoding through the public code representation, without 
knowing the private code representation.
The most powerful algorithms that can be used for this purpose are \ac{ISD} 
algorithms. These algorithms aim at performing decoding of any linear block 
code by exploiting a general representation of it. \ac{ISD} algorithms have 
been introduced by Prange \cite{Prange1962} and subsequently improved by 
Lee-Brickell \cite{Lee1988}, Leon \cite{Leon1988} and Stern \cite{Stern1989}. 
More recently, they have known great advances through modern approaches, also 
exploiting the generalized birthday paradox 
\cite{May2011,Peters2010,Bernstein2011,Becker2012,Niebuhr2017}. It is possible 
to show that the general decoding problem is equivalent to the problem of finding 
low-weight codewords in a general (random-like) code. 
Therefore, algorithms for searching low-weight codewords can be used as \ac{ISD} 
algorithms.

The availability of an efficient algorithm to search for low-weight codewords is 
also at the basis of \acp{KRA}. In \sysacro{} the matrix $L=HQ$ is a valid 
parity-check matrix for the public code. Since $L$ is sparse, by knowing it 
an attacker could separate $H$ from $Q$ and recover the secret key. In order to 
discover $L$, an attacker must search for its rows in the dual of the public 
code. Due to the sparsity of $H$ and $Q$, any of these rows has weight in the 
order of $n_0 d_v m$. The attack can be implemented by exploiting again an efficient 
algorithm for the search of low-weight codewords in linear block codes.

Another potential attack to systems based on \ac{QC-LDPC} codes is that 
presented in \cite{Shooshtari2016}. This attack uses a special squaring 
technique and, by extracting the low-weight error vectors, finds low-weight 
codewords more efficiently than with a general \ac{ISD} algorithm. 
This attack, however, is applicable if and only if $p$ is even.
Therefore, in order to increase the system security it is advisable to choose 
odd values of $p$.
Choosing $p$ as a prime is an even more conservative choice against 
cryptanalysis exploiting factorization of $p$.
The value of $p$ in \sysacro{} is chosen in such a way to prevent these attacks.

To estimate complexity of \acp{DA} and \acp{KRA} exploiting \ac{ISD} 
and low-weight codeword searching algorithms, let us define the \ac{WF} of an 
algorithm as the base-2 logarithm of the average number of binary operations it 
requires to complete its execution successfully.
Let $WF(n,k,w)$ denote the \ac{WF} of the most efficient algorithm searching 
for codewords of weight $w$ in a code having length $n$ and dimension $k$.
Such an algorithm can be used to perform \ac{ISD} with the aim of decrypting a 
\sysacro{} ciphertext without knowing the private key.
In this case, we have $n=n_0p$, $k=(n_0-1)p$ and $w=t$. Moreover, due to the 
\ac{QC} nature of the codes, a speedup in the order of $\sqrt{p}$ must be taken 
into account \cite{Sendrier2011}.
Hence, the security level against decoding attacks of this type can be computed as
\begin{equation}
SL_{DA} = \frac{WF(n_0p,(n_0-1)p,t)}{\sqrt{p}}.
\label{eq:WFdec}
\end{equation}

Concerning the \acp{KRA} attack, based on the above considerations we have a 
similar formula, but with different parameters, that is,
\begin{equation}
SL_{KRA} = \frac{WF(n_0p,p,n_0 d_v m)}{p},
\label{eq:WFkra}
\end{equation}
where the speedup factor $p$ is due to the fact that recovering only one out of $p$ sparse rows of $L$, 
 is enough for the attacker (due to the \ac{QC} 
structure of $L$).

According to \cite{Vries2016}, the most efficient \ac{ISD} algorithm taking into 
account Grover's algorithm \cite{Grover1996} running on a quantum computer is 
Stern's algorithm.
Therefore, the post-quantum security levels have been estimated by considering 
the work factor of Stern's algorithm with quantum speedup according to \cite{Vries2016}.
Instead, with classical computers the most efficient \ac{ISD} algorithm turns 
out to be the BJMM algorithm in \cite{Becker2012}. Therefore, the security levels 
against attackers provided with classical computers have been estimated by 
considering the work factor of BJMM in \eqref{eq:WFdec} and 
\eqref{eq:WFkra}. 
\mb{We chose to employ the results provided in~\cite{Vries2016} to evaluate
the computational efforts of Stern's variant of the \ac{ISD} as they provide
exact formulas instead of asymptotic bounds.
However, we note that a recent work~\cite{KT17} provides improved asymptotic 
bounds on the computational complexity of quantum \ac{ISD} for increasing
values of the codeword length $n$. Deriving from this approach exact values for given parameters set is worth investigating.}

\subsection{System parameters}\label{subsec: Syspar}
The NIST call for Post-Quantum Cryptography Standardization \cite{NISTcall2016} 
defines $5$ security categories, numbered from $1$ to $5$ and characterized by 
increasing strength (see \cite{NISTcall2016} for details).
According to this classification, nine instances of \sysacro{} are proposed, grouped in three classes corresponding to different security levels.
The three instances in each class correspond to three values of $n_0$ ($2, 3, 4$), 
each one yielding a different balance between performance and public key size.
The parameters of the nine instances of \sysacro{} are reported in 
Table~\ref{tab:sys_parameters} for the security categories $1$, $3$ and $5$, 
respectively.
In the table, the superscript (pq) denotes that the attack work factor has been 
computed taking into account quantum speedups due to Grover's algorithm, while 
the superscript (cl) denotes that only classical computers have been considered.

For each security category and considered value of $n_0$, we have fixed a value 
of the parity-check matrix row/column weight $d_v$ in the order of $25$ or 
less (that is advisable to have good error correcting capability of the private 
\ac{QC-LDPC} codes), and we have found the values of $p$ and $m$ that allow satisfying \eqref{eq:WFkra} for the target security level.
\mb{In fact, the value of $m$ must be chosen such that the dual of the public code, having minimum distance equal to $n_0md_v$, is robust against \acp{KRA} based on \ac{ISD}.
Once $n_0$ is fixed, we can find many pairs of values $m$ and $d_v$ which satisfy this bound; among these, we have chosen the one having the lowest product $md_v$, which is a metric affecting the error correcting capability of the private code.
}
Then, we have found the value of $t$ that allows satisfying \eqref{eq:WFdec} 
and checked whether $t' = tm$ errors can be corrected by the private code 
through Q-decoding with a sufficiently low \ac{DFR}.
Otherwise, we have increased the value of $p$ keeping all the other parameters 
fixed.
Concerning the estimation of the \ac{DFR}, we have first exploited \ac{BF} 
asymptotic thresholds \cite{Baldi2012}, and then we have performed Montecarlo 
simulations for each system instance in order to evaluate its \ac{DFR}.
In all Montecarlo simulations, except the one for the Category 1, $n_0=2$ parameter
set, we have encountered no errors, so the \ac{DFR} can 
be approximately bounded by the reciprocal of the number of simulated decryptions.
Concerning the parameter set for Category 1, $n_0 = 2$, we obtained $20$ 
failures on $2.394\cdot 10^9$ decoding computations, pointing to a DFR $\approx 8.3\cdot 10^{-9}$.\newline \noindent
In order to make a conservative design of the system, we have considered some 
margin in the complexity estimates of the attacks, such that the actual security 
level for these instances is larger than the target one.
This also accounts for possible (though rare) cancellations occurring in $L$, 
which may yield a row weight slightly smaller than $m d_v n_0$.
The values of $d_v$ have been chosen greater than $15$ in order to avoid codes having too small minimum distances.
In addition, they are odd to ensure that the circulant blocks forming $H$ and $L$ 
(and $L_{n_0-1}$, in particular) have full rank.
Also the values of $m$ are always odd, and the sets $[m_0,m_1,\cdots,m_{n_0-1}]$ 
have been chosen in such a way to guarantee that $Q$ has full rank.
In fact, $L=HQ$ is a valid parity-check matrix for the public code: if $Q$ is 
singular, it might happen that the rank of $L$ is lower than $p$, leading to a 
code with a co-dimension lower than $p$.
With the choice of an invertible $Q$, we guarantee that this does not occur.

\begin{table}[!t]
\renewcommand{\arraystretch}{1.3}
\caption{Parameters for \sysacro{} and  estimated computational efforts to break a given instance as a function of the security category and number of circulant blocks $n_0$\label{tab:sys_parameters}
}
\centering
\footnotesize
\resizebox{\textwidth}{!}{%
\begin{tabular}{cc|ccccccccc}
\hline
\textbf{Category} & $\mathbf{n_0}$ & $\mathbf{p}$     & $\mathbf{d_v}$ & $\mathbf{[m_0,\cdots, m_{n_0-1}]}$  & $\mathbf{t}$   & $\mathbf{SL^{(pq)}_{DA}}$ & $\mathbf{SL^{(pq)}_{KRA}}$ & $\mathbf{SL^{(cl)}_{DA}}$ & $\mathbf{SL^{(cl)}_{KRA}}$ & $\mathbf{DFR}$    \\
\hline
\multirow{3}{*}{$1$} & $2$   & $27,779$ & $17$  & $[4,3]$                  & $224$ & $135.43$         & $134.84$          & $217.45$         & $223.66$          & $\approx$$8.3$$\cdot$$10^{-9}$\\
                     & $3$   & $18,701$ & $19$  & $[3,2,2]$                & $141$ & $135.63$         & $133.06$          & $216.42$         & $219.84$          & $\lesssim 10^{-9}$\\
                     & $4$   & $17,027$ & $21$  & $[4,1,1,1]$              & $112$ & $136.11$         & $139.29$          & $216.86$         & $230.61$          & $\lesssim 10^{-9}$\\
\hline                     
\multirow{3}{*}{$2$--$3$} & $2$ & $57,557$ & $17$ & $[6,5]$ & $349$ & $200.47$ & $204.84$ & $341.52$ & $358.16$ & $\lesssim 10^{-8}$    \\
                          & $3$ & $41,507$ & $19$ & $[3,4,4]$ & $220$ & $200.44$ & $200.95$ & $341.61$ & $351.57$ & $\lesssim 10^{-8}$  \\
                          & $4$ & $35,027$ & $17$ & $[4,3,3,3]$ & $175$ & $200.41$ & $201.40$ & $343.36$ & $351.96$ & $\lesssim 10^{-8}$\\
\hline                                           
\multirow{3}{*}{$4$--$5$} & $2$ & $99,053$ & $19$ & $[7,6]$ & $474$ & $265.38$ & $267.00$ & $467.24$ & $478.67$ & $\lesssim 10^{-8}$    \\                         
                          & $3$ & $72,019$ & $19$ & $[7,4,4]$ & $301$ & $265.70$ & $270.18$ & $471.67$ & $484.48$ & $\lesssim 10^{-8}$  \\
                          & $4$ & $60,509$ & $23$ & $[4,3,3,3]$ & $239$ & $265.48$ & $268.03$ & $473.38$ & $480.73$ & $\lesssim 10^{-8}$\\
\hline
\end{tabular}
}
\end{table}

\section{Properties of the proposed cryptosystem \label{sec:Properties}}
\mb{
The \ac{QC-LDPC} code-based Niederreiter cryptosystem alone achieves only \ac{IND-CPA}, that however is sufficient in case of using ephemeral keys.
It is possible to convert a Niederreiter cryptosystem achieving only \ac{IND-CPA} into one achieving \ac{IND-CCA}, under the assumption that the \ac{DFR} of the underlying code is zero.
Such a conversion involves substituting the outcome of a decoding failure (due to an ill-formed ciphertext) with the outcome of a \ac{KDF} taking as input either the public syndrome and a fixed secret bit sequence~\cite{cryptoeprint:2017:604,cryptoeprint:2017:1005}, or a secret permutation of the syndrome itself~\cite{Cayrel2017}.
We apply the conversion specified in~\cite{cryptoeprint:2017:604} to our scheme, despite its \ac{DFR} is not null, as it still proves beneficial in case of an accidental keypair reuse, against an attacker matching the \ac{IND-CCA} model whenever no decoding failures due to the \ac{QC-LDPC} code structure takes place.
}
Furthermore, we note that \sysacro{} ciphertexts are not malleable in a chosen plaintext scenario.
Indeed, even if an attacker alters arbitrarily a ciphertext so that it decrypts to a valid 
error vector $e$  (e.g., discarding the ciphertext  and 
forging a new one), the shared secret is derived via a hash based \ac{KDF}, 
which prevents him from controlling the output of the decryption.
\subsubsection{Relations with the security of \ac{QC-MDPC} code-based systems.}
\mb{Differently from \ac{QC-MDPC} code-based systems, the public code in \sysacro{} has a \ac{QC-MDPC} matrix $L$ that can be factorized into $H$ and $Q$, and this might appear to yielding lower security than a general \ac{QC-MDPC} matrix.
However, in order to attempt factorization of $L$, the attacker should first recover it by searching 
for low-weight codewords in the dual of the public code.
Once $L$ has been recovered, trying to factorize it into $H$ and $Q$ indeed becomes pointless, since 
the attacker could exploit $L$ to perform direct decoding of the public \ac{QC-MDPC} code.
Alternatively, an attacker could try to perform decoding of the public code, which requires solving the syndrome decoding problem for the same code.
The best known techniques for solving these two problems are based on \ac{ISD}, and no method is known to facilitate their solution by exploiting the fact that $L$ can be factorized into $H$ and $Q$.
}
\subsubsection{Risks in case of keypair reuse.}
While \sysacro{} uses ephemeral keys that are meant for single use, 
it is possible that implementation accidents lead to a reuse of the 
same keypair more than once.
The main threat in case of keypair reuse is the reaction attack described 
in~\cite{Fabsic2017}, where a correlation between the \ac{DFR} and the private key is derived.
However, for the attack to succeed, the attacker needs to reliably estimate 
the decoding failure rate for a set of carefully crafted or selected error vectors.
Given the DFR for which \sysacro{} was designed ($<10^{-8}$), obtaining a 
reliable estimate requires a number of decryptions with the same key in the 
order of billions. Since the said evaluation should be obtained for all the 
possible distances between two set bits in the secret key, a conservative 
estimate of the number of decryption actions required is $(p-1) \frac{1}{DFR}$, which, 
considering the weakest case, corresponding to Category $1$ with $n_0=2$,
yields $\gtrsim 2.7 \times 10^{12}$ decryptions.
Therefore, the attack presented in~\cite{Fabsic2017} is not a practical threat 
 on \sysacro{} with the proposed parameters, unless a significant amount of 
 decryptions are performed with the same key.
\mb{Moreover, even the \ac{CCA} described in \cite{Cayrel2017}, where a ciphertext
is crafted with a number of errors greater than $t$ to artificially increase the \ac{DFR} of the system,
can be thwarted through checking the weight of the decoded error vector and reporting 
a decoding failure if it exceeds $t$.}

\subsubsection{Protection against side-channel attacks.}
The two most common side channels exploited to breach practical implementations 
of cryptosystems are the execution time of the primitive and the instantaneous 
power consumption during its computation.
In particular, in~\cite{Fabsic2016}, it was shown how a \ac{QC-LDPC} code-based 
system can be broken by means of simple power analysis, exploiting the 
control-flow dependent differences of the decoding algorithm.
We note that employing ephemeral keys provides a natural resistance 
against non-profiled power consumption side channel attacks, as a significant 
amount of measurements with the same key ($>30$) must be collected before the 
key is revealed. 

Concerning execution time side channel information leakage, the main portion of
the \sysacro{} decryption algorithm which is not characterized by a constant execution time is decoding.
Indeed, the number of iterations made by the decoder depends on the values 
being processed.
However, for the proposed parameters, we note that the number of iterations is 
between 3 and 5, with a significant bias towards 4.
Hence, it is simple to achieve a constant time decoding by modifying the 
algorithm so that it always runs for the maximum needed amount of iterations to 
achieve the desired DFR. Such a choice completely eliminates the timing leakage,
albeit trading it off for a performance penalty.
\begin{table}[!t]
\small
\begin{center}
\caption{Running times for key generation, encryption and decryption as a 
function of the category and the number of circulant blocks $n_0$ on an 
AMD Ryzen 5 1600 CPU.\label{tab:runningtimes}}
\begin{tabular}{cc|cccc}
\toprule
\multirow{2}{*}{Category} & \multirow{2}{*}{$n_0$}  &  KeyGen  & Encrypt & Decrypt              & Total CPU time  \\
                          &                         &  (ms)    &  (ms)   &     (ms)             & Ephemeral KEM (ms)  \\
\midrule
\multirow{3}{*}{1}        & $2$ & $34.11$  ($\pm$$1.07$) & $2.11$   ($\pm$$0.08$) & $16.78$   ($\pm$$0.53$)  &   $52.99$  \\   
                          & $3$ & $16.02$  ($\pm$$0.26$) & $2.15$   ($\pm$$0.17$) & $21.65$   ($\pm$$1.71$)  &   $39.81$  \\   
                          & $4$ & $13.41$  ($\pm$$0.23$) & $2.42$   ($\pm$$0.08$) & $24.31$   ($\pm$$0.86$)  &   $40.14$  \\
\midrule                  
\multirow{3}{*}{2--3}     & $2$ & $142.71$ ($\pm$$1.52$) & $8.11$   ($\pm$$0.21$) & $48.23$   ($\pm$$2.93$)  &   $199.05$  \\   
                          & $3$ & $76.74$  ($\pm$$0.78$) & $8.79$   ($\pm$$0.20$) & $49.15$   ($\pm$$2.20$)  &   $134.68$  \\   
                          & $4$ & $54.93$  ($\pm$$0.84$) & $9.46$   ($\pm$$0.28$) & $46.16$   ($\pm$$2.03$)  &   $110.55$  \\  
\midrule                  
\multirow{3}{*}{4--5}     & $2$ & $427.38$  ($\pm$$5.15$) & $23.00$   ($\pm$$0.33$) & $91.78$   ($\pm$$5.38$)  &   $542.16$  \\   
                          & $3$ & $227.71$  ($\pm$$1.71$) & $24.85$   ($\pm$$0.37$) & $92.42$   ($\pm$$4.50$)  &   $344.99$  \\   
                          & $4$ & $162.34$  ($\pm$$2.39$) & $26.30$   ($\pm$$0.53$) & $127.16$  ($\pm$$4.42$)  &   $315.80$  \\ 
\bottomrule
 \end{tabular}
 \end{center}
\end{table}

\section{Implementation and numerical results}
\label{sec:results}

An effort has been made to realize a fast and efficient C99 implementation of \sysacro{} without platform-dependent optimizations, \mb{which is publicly available in \cite{LEDAcrypt}}.
To this end, we represented each circulant block as a polynomial in 
$\mathbb{F}_2[x]/\langle x^p +1 \rangle$ thanks to the isomorphism described in 
Section \ref{subsec:prel}.
Consequently, all the involved block circulant matrices are represented as
matrices of polynomials in $\mathbb{F}_2[x]/\langle x^p +1 \rangle$.
The polynomials are materialized employing a bit-packed form of their binary
coefficients in all the cases where the number of non null coefficients is high.
In case a polynomial has a low number of non null coefficients with respect to the 
maximum possible, i.e., the circulant matrix is sparse, 
we materialize only the positions of its one coefficients
as integers.

We provide below the results of a set of execution time benchmarks.
The results were obtained measuring the required time for key generation, 
encryption (key encapsulation) and decryption (key decapsulation) as a function 
of the chosen security category and the number of circulant blocks $n_0$. 
The measurements reported are obtained as the average of $100$ executions of the 
reference implementation.
The generated binaries were run on an AMD Ryzen 5 1600 CPU at 3.2 GHz, locking the 
frequency scaling to the top frequency.

Table~\ref{tab:runningtimes} reports the running times in terms of
CPU time taken by the process. As it can be noticed, the most computationally
demanding primitive is the key generation, which has more than $80$\% of its
computation time taken by the execution of a single modular inverse in 
$\mathbb{F}_2[x]/\langle x^p+1 \rangle$ required to obtain the value of 
$L_{n_0-1}^{-1}$.
The encryption primitive is the fastest among all, and its computation time is
substantially entirely devoted ($> 99\%$) to the $n_0-1$ polynomial multiplications
performing the encryption.
The decryption primitive computation is dominated by the Q-decoder computation 
($>95\%$ of the time), with a minimal portion taken by the $n_0$ modular 
multiplications which reconstruct $L_{n_0-1}$ and the one to compute
the private syndrome fed into the Q-decoder.

Considering the computational cost of performing a KEM with ephemeral keys, 
the most advantageous choice is to pick $n_0=4$ for any security level, although
the computational savings are more significant when considering high-security 
parameter choices (Category 3 and 5).

\begin{table}[!t]
\small
\begin{center}
\caption{Sizes of the keypair and encapsulated shared secret as a function of 
the chosen category and number of circulant blocks 
$n_0$. \label{tab:keysizes}}
 \begin{tabular}{cc|ccccc}
 \toprule
\multirow{2}{*}{Category} & \multirow{2}{*}{$n_0$}  & \multicolumn{2}{c}{\underline{Private Key Size (B)}} & Public Key  &  Shared secret  & Enc secret  \\
                          &                         &  At rest       &       In memory         &  size (B)   &  size (B)       & size (B)\\
\midrule                             
\multirow{3}{*}{1}    &  $2$ &  $24$ &  $  668$ & $ 3,480$ &  $3,480$ & $32$ \\ 
                      &  $3$ &  $24$ &  $  844$ & $ 4,688$ &  $2,344$ & $32$ \\ 
                      &  $4$ &  $24$ &  $1,036$ & $ 6,408$ &  $2,136$ & $32$ \\
\midrule                      
\multirow{3}{*}{2--3} &  $2$ &  $32$ &  $  972$ & $ 7,200$ &  $7,200$ & $48$ \\
                      &  $3$ &  $32$ &  $1,196$ & $10,384$ &  $5,192$ & $48$ \\
                      &  $4$ &  $32$ &  $1,364$ & $13,152$ &  $4,384$ & $48$ \\
\midrule                      
\multirow{3}{*}{4--5} &  $2$ &  $40$ &  $1,244$ & $12,384$ & $12,384$ & $64$ \\
                      &  $3$ &  $40$ &  $1,548$ & $18,016$ & $ 9,008$ & $64$ \\
                      &  $4$ &  $40$ &  $1,772$ & $22,704$ & $ 7,568$ & $64$ \\ 
 \bottomrule
 \end{tabular}
\end{center}
\end{table}

Table~\ref{tab:keysizes} reports the sizes of both the keypairs and the 
encapsulated secrets for \sysacro{}.
In particular, regarding the size of the private keys we report
both the size of the stored private key
and the required amount of main memory to store the expanded
key during the decryption phase. 
We note that, for a given security category, increasing the value of $n_0$
enlarges the public key, as it is constituted of $(n_0-1)p$ bits.
This increase in the size of the public key represents a tradeoff with the
decrease of the size of the ciphertext to be transmitted since it is only 
$p$ bits long, and $p$ decreases if a larger number of blocks is selected, for
a fixed security category.
The size of the derived encapsulated secret is at least $256$ bits, in order to
meet the requirement reported in \cite{NISTcall2016}.
The shared secret is derived employing the SHA-$3$ hash function with 
a $256$, $384$ or $512$ bits digest, in order to match the requirements of Categories $1$, $3$, and $5$, respectively.

\section{Conclusion}
\label{sec:conclusion}

We have introduced a post-quantum \ac{KEM} based on \ac{QC-LDPC} codes with the following advantages: it is built on an NP-complete problem under reasonable assumptions; it exploits improved \ac{BF} decoders which are faster than classical \ac{BF} decoders; it requires compact keypairs (below 23 kiB at most), with minimum size private keys; it needs only addition and multiplication over $\mathbb{F}_2[x]$, and modular inverse over $\mathbb{F}_2[x]/\langle x^p+1 \rangle$ besides single-precision integer operations; it is particularly efficient in applying countermeasures against non-profiled power consumption side channel attacks.
As regards implementation, no platform specific optimizations have been exploited, thus we expect these results to be quite consistent across different platforms. On the other hand, starting from this platform-agnostic reference implementation, a number of optimizations can be applied to make \sysacro{} faster.

\section*{Acknowledgments}
Paolo Santini was partly funded by Namirial SpA.

\bibliographystyle{splncs03}
\bibliography{Archive}
\end{document}